\documentclass[11pt]{article}

\usepackage{palatino}
\usepackage{mathpazo}
\usepackage{graphicx}
\usepackage{color}
\usepackage[margin=1in]{geometry}
\usepackage[draft=false]{hyperref}

\usepackage{amssymb, amsthm, amsmath}
\usepackage[linesnumbered,boxed,ruled,vlined]{algorithm2e}
\newtheorem{theorem}{Theorem}
\newtheorem{lemma}[theorem]{Lemma}
\newtheorem{Def}{Definition}

\newtheorem{Exam}{Example}

\newtheorem{fact}[theorem]{Fact}

\def\<{\langle}
\def\>{\rangle}
\def\be{\begin{equation*}}
\def\ee{\end{equation*}}
\def\bea{\begin{eqnarray*}}
\def\eea{\end{eqnarray*}}

\newcommand{\comment}[1]{}

\newcommand{\tr}{\operatorname{Tr}}

\newcommand{\cH}{\mathcal{H}}

\newcommand{\cD}{\mathcal{D}}

\newcommand{\ip}[2]{\langle #1 , #2\rangle}

\newcommand{\ket}[1]{\ensuremath{\left|#1\right\rangle}}
\newcommand{\op}[2]{|#1\rangle \langle #2|}
\newcommand{\Tr}{\mathop{\mathrm{tr}}\nolimits}

\newcommand{\setft}[1]{\mathrm{#1}}

\def\H{\mathcal{H}}

\newcommand{\lin}[1]{\setft{L}\left(#1\right)}
\title{\bf\LARGE Quantum Earth mover's distance, No-go Quantum Kantorovich-Rubinstein theorem, and Quantum Marginal Problem}

\author{
Nengkun Yu\footnote{%
    Centre for Quantum Software and Information, School of Software, Faculty of Engineering and Information
Technology, University of Technology Sydney, NSW, Australia}
  \and
  Li Zhou\footnote{
    Department of Computer Science and Technology, Tsinghua University, Beijing, China; and Centre for Quantum Software and Information, School of Software, Faculty of Engineering and Information
Technology, University of Technology Sydney, NSW, Australia}
\and
  Shenggang Ying\footnote{
  Centre for Quantum Software and Information, School of Software, Faculty of Engineering and Information
Technology, University of Technology Sydney, NSW, Australia
  }
 \and
 Mingsheng Ying\footnote{Centre for Quantum Software and Information, School of Software, Faculty of Engineering and Information
Technology, University of Technology Sydney, NSW, Australia; and Institute of Software, Chinese Academy of Sciences, Beijing 100190, China; and Department of Computer Science and Technology, Tsinghua University, Beijing 100084, China
  }
}

\date{\today}  % Fix this later

\begin{document}
\maketitle

%---------Abstract---------%
\begin{abstract}
The earth mover's distance is a measure of the distance between two probabilistic measures. It plays a fundamental role in mathematics and computer science. The Kantorovich-Rubinstein theorem provides a formula for
the earth mover's distance on the space of regular probability
Borel measures on a compact metric space. In this paper, we investigate the quantum earth mover's distance. We show a no-go Kantorovich-Rubinstein theorem in the quantum setting. More precisely, we show that the trace distance between two quantum states can not be determined by their earth mover's distance. The technique here is to track the bipartite quantum marginal problem. Then we provide inequality to describe the structure of quantum coupling, which can be regarded as quantum generalization of Kantorovich-Rubinstein theorem. After that, we generalize it to obtain into the tripartite version, and build a new class of necessary criteria for the tripartite marginal problem.
\end{abstract}

%%%%%%%%%%%%%%%%%%%%%%%%%%%%%%%%%%%%%%%%%%%%%%%%%%%%%%%%%%%%%%%%%%%%%%
\section{Introduction}
In mathematics and economics, transportation theory studied the optimal transportation and allocation of resources. Gaspard Monge formalized it in 1781 \cite{Monge1781}. Leonid Kantorovich, the Soviet mathematician and economist, made major advances in the field during World War II \cite{Kantorovich1942}. The central concept in transportation theory is the earth mover's distance, also known as the Wasserstein metric, which measures the distance between two probabilistic distributions over a region. Intuitively, given two distributions, one can be seen as a mass of earth properly spread in space, the other as a collection of holes in that same space. Then, the earth mover's distance measures the least amount of work needed to fill the holes with earth. Here, a unit of work corresponds to transporting a unit of earth by a unit of ground distance. The definition of the earth mover's distance is valid only if the two distribution have the same integral(the capacity of the holes equals the amount of the mass), as in normalized probabilistic density functions. The celebrated Kantorovich-Rubinstein theorem \cite{Kantorovich1958,Kantorovich1982} characterizes the earth mover's distance by considering all joint distributions whose marginal distributions are the two probabilistic distributions, the so called probabilistic  coupling.

Besides being well-studied in probability theory and
the theory of optimal transport, the earth mover's distance is increasingly seeing applications in computer science and beyond.
It is widely used in content-based image retrieval to compute distances between the color histograms of two digital images \cite{Peleg1989}.In this case, the region is the image's domain, and the total amount of light (or ink) is the dirt to be rearranged. The same technique can be used for any other quantitative pixel attribute, such as luminance, gradient, apparent motion in a video frame, etc. An optimal transportation model is introduce in studying domain adaptation in \cite{CFTR17}.
More generally, the Earth mover's distance is used in pattern recognition to compare generic summaries or surrogates of data records called signatures, see a recent survey \cite{PC19}. Very recently, this earth mover's distance is introduced to study the probabilistic programming \cite{Barth2018} as well as Generative Adversarial Networks of machine learning \cite{Arjovsky2017}.

Back to quantum information science, the concept of quantum states, a quantum counter part of probabilistic distribution, plays significant role in quantum information science as it is used for carrying quantum information in information processing tasks.
We consider finite-dimensional complex Hilbert spaces as quantum state spaces. A pure state is represented by a unit vector in such a space. General quantum states, the so-called density matrices or mixed states, are described by positive
semidefinite matrices with unit trace. Since the multipartite Hilbert space is the tensor product of the individual particles' Hilbert spaces, the dimension of the multipartite Hilbert space scales exponentially in the number of particles. Although it is highly desired to understand the behavior of quantum systems via classical modeling, this exponential behavior becomes one of the main obstructions.

Interestingly, it has been widely known that many important physical quantities, for instance energy and entropy, depend on very small parts of the whole system only, $i.e.$, the so-called marginal or reduced density matrices. On the other hand, in reality, quantum states of many physically realistic quantum systems usually involve only few-body interactions \cite{hastings2010locality,wolf2008area,perez2006matrix,verstraete2008matrix,cirac2009renormalization,Schilling14,Schilling17}. These observations can actually save
lots of parameters. As an example, one can observe that for the Hamiltonians arose from quantum chemistry which contain at most
2-body interactions, the number of free
parameters scales at most quadratically
in the number of particles. Coulson \cite{Coulson1960,Tredgold1957} proposed the following
problem: How to characterize the allowed sets of 2-body correlations or
density operators between all pairs of $N$ particles?
The general problem of characterizing the set of possible reduced density matrices (maybe $k$-local),
known as the quantum marginal problem, has been considered as
one of the most fundamental problems in quantum information theory and in quantum chemistry \cite{Coleman1963,Ruskai1969,CY2000,Stillinger1995}. A very
large effort has been devoted to understanding this problem \cite{Cioslowski2000,Mazziotti2006}, and it is proved to be NP-hard and QMA (quantum Merlin-Arthur) complete in \cite{Liu2006,Liu2007}.

Due to the significance of the quantum marginal problem, there are many attempts in understanding it, even in the case of low dimension or a small number of parties. Bravyi characterizes the relation between the spectra of local one-qubit states and the spectra of the whole two-qubit mixed states by a remarkable
explicit argument \cite{Bravyi2004}. In \cite{Chen2014}, the two-qubit symmetric extension problem, a special marginal problem, is completely solved. This result provides the first analytic necessary and sufficient condition
for the quantum marginal problem with overlapping marginal. For the general marginal problem, some necessary and some sufficient conditions are provided in \cite{Osborne2008,Carlen2013} by using the celebrated Strong Subadditivity of entropy.

One induced problem also has attracted a lot of attention: Characterizing of the one-body reduced
density matrices of a pure global quantum state \cite{Klyachko2004,Daftuar2004,Klyachko2006,Christandl2006}. For multiqubit case, this problem is completely solved by
Higuchi, Sudbery and Szulc by completely determining the possible
one-qubit reduced states in\cite{Higuchi2003}. In \cite{Yu2013}, it is proved that multipartite W-type state, a special class of multi-qubit states, is determined by its single-particle reduced density matrices among all W-type states. For general multipartite pure state, \cite{Christandl2014} proposed an efficient method to compute the joint probability distribution of the eigenvalues of its one-body reduced density
matrices.

In this paper, we study the possible generalization of the earth mover's distance. In particular, we show a no-go quantum Kantorovich-Rubinstein theorem for almost all quantum distance. Our main technique is to study the quantum coupling, or equivalently, the bipartite quantum marginal problem. Then, the idea is used to study the tripartite quantum marginal problem. The structure of this paper is as follows.

In Section 2, we first provide some notations and preliminaries of the distance and fidelity of quantum states. After that, the basic definition and examples of the quantum marginal problem are given.

In Section 3, we propose a quantum version of the earth mover's distance. Then we show a no-go Kantorovich-Rubinstein theorem in the quantum setting. More precisely, we show that the trace distance between two quantum states can not be determined by their earth mover's distance. Our main technique is to
study the bipartite quantum marginal problem. We observe that the fidelity between the two marginal states is at least the distance between the two probabilities obtained by measuring the bipartite state through the projective measurements onto the symmetric subspace and anti-symmetric subspace.

In Section 4, we study the largest overlap of quantum coupling and projection onto the symmetric subspace for given quantum marginals. We obtain two lower bounds, one for diagonal marginals, the other for general marginals. Our result can be regarded as quantum generalization of Kantorovich-Rubinstein theorem, although in the inequality fashion rather than equality.

In Section 5, we study the tripartite marginal problems. For tripartite state $\rho_{ABC}$ with the dimensions of $A$ and $B$ being equal, we define the two probability distributions which depend on $\rho_{AB}$. Then, we show that the distance between the two marginal states $\rho_{AC}$ and $\rho_{BC}$ is at most the fidelity between the two probability distributions. On the other hand, the fidelity between $\rho_{AC}$ and $\rho_{BC}$ is at least the distance between the two probability distributions.
By changing the local operation on subsystem $A$, we are able to provide a class of necessary criteria for the tripartite marginal problem.

In Section 6, we mention some open questions regarding this marginal problem.
%%%%%%%%%%%%%%%%%%%%%%%%%%%%%%%%%%%%%%%%%%%%%%%%%%%%%%%%%%%%%%%%%%%%%%
\section{Background}
\subsection{Notations and Preliminaries}
We use the symbol $\cH$ to denote the finite dimensional Hilbert space over complex numbers, $d_{\cH}$ to denote its dimension and
$\lin{\cH}$ to denote the set of linear operators mapping from $\cH$ into itself.
Let $\mathrm{Pos}(\cH)\subset\lin{\cH}$ be the set of positive (semidefinite) matrices, and $\cD(\cH) \subset\mathrm{Pos}(\cH)$ is the set of positive (semidefinite) matrices with trace one.
A pure quantum state of $\cH$ is a normalized vector $\ket{\psi}\in\cH$, while a general quantum state is characterized by density operator $\rho\in\cD(\cH)$.
For simplicity, we use $\psi$ to represent the density operator of a pure state $\ket{\psi}$ which is just the projector $\psi = \op{\psi}{\psi}$.
A density operator $\rho$ can always be decomposed into a convex combination of pure states:
\be
\rho = \sum_{k=1}^{n}p_k\op{\psi_k}{\psi_k},
\ee
where the coefficients $p_k$ are positive numbers and add up to one.

To characterize the difference between the quantum states, there are two commonly used measures: trace distance and fidelity. The trace distance $D$ between two density operators $\rho$ and $\sigma$ is defined as
\be
D(\rho,\sigma) \equiv \frac{1}{2}||\rho-\sigma||_1
\ee
where we define $||A||_1\equiv\Tr\sqrt{A^\dag A}$ to be the trace of the positive square root of $A^\dag A$.

Notice that this is a direct generalization of the distance between classical distributions, usually called total variance distance.

We use $||A||_2\equiv\sqrt{\Tr A^\dag A}$ to be denote the $2$-norm of $A$.

The matrix $1$-norm satisfies the following triangle inequalities and H{\"o}lder's inequality (Cauchy inequality).
%I add the second inequality as we used in proof of Thm 7, page 10.
\begin{fact}
\begin{align*}
||A+B||_1&\leq ||A||_1+||B||_1.\\
||A-B||_1&\geq |||A||_1-||B||_1|.\\
||AB||_1&\leq ||A||_2||B||_2.
\end{align*}
\end{fact}
It is direct to verify the following strong concavity statement about the distance between the mixture of quantum states by the triangle inequality given above.
\begin{fact}\label{strong concavity}
For quantum states $\rho_i$, $\sigma_i$ and probability distribution $(p_0,p_1,\cdots,p_n)$
\be
D\left( \sum_ip_i\rho_i,\sum_ip_i\sigma_i \right) \leq \sum_{i=0}^n p_i D(\rho_i,\sigma_i).
\ee
\end{fact}

The fidelity between states $\rho$ and $\sigma$ is defined to be
\be
F(\rho,\sigma) \equiv \mathrm{Tr}\sqrt{\sqrt{\rho}\sigma\sqrt{\rho}}.
\ee
For pure states $\ket{\psi}$ and $\ket{\phi}$, $F(\psi,\phi)=|\ip{\psi}{\phi}|$.

The strong concavity property for the fidelity (\cite{NielsenC00}) can be formalized as
\begin{fact}\label{strong concavity}
For quantum states $\rho_i$, $\sigma_i$ and probability distributions $(p_0,p_1,\cdots,p_n)$ and $(q_0,q_1,\cdots,q_n)$
\be
F\left( \sum_ip_i\rho_i,\sum_iq_i\sigma_i \right) \ge \sum_{i=0}^n\sqrt{p_iq_i}F(\rho_i,\sigma_i).
\ee
\end{fact}

\begin{Def}\label{def:purification}
	We say that a pure state $\ket{\psi} \in \H_A\otimes\H_B$
	is a purification of some state $\rho$
	if $\Tr_A({\psi})=\rho$.
\end{Def}
In the above definition, $\Tr_A(\psi)$ means that we trace out the subsystem $A$ from $\psi$. Its formal definition is $\Tr_A(\psi) = \sum_i \<i|_A \psi |i\>_A$.

\begin{fact}[Uhlmann's theorem, \cite{Uhl77}]
	\label{fac:Uhlmann}
	Given quantum states $\rho$, $\sigma$, and a purification $\ket{\psi}$ of $\rho$,
	it holds that $F({\rho},{\sigma})=\max\limits_{\ket{\phi}} | \langle \phi | \psi \rangle | $,
	where the maximum is taken over all purifications of $\sigma$.
\end{fact}

\begin{fact}
	\label{fact:fdrelation}
	Given quantum states $\rho$, $\sigma$,
$$
F^2({\rho},{\sigma})+D^2({\rho},{\sigma})\leq 1.
$$
\end{fact}

\subsection{Marginal Problem}
In this subsection, we provide the definition and some examples of quantum marginal problem. We consider the following multipartite Hilbert space
\begin{equation*}
\mathcal{H}_{S}=\mathcal{H}_A\otimes\mathcal{H}_B\otimes\mathcal{H}_C\otimes\cdots
\end{equation*}
where $\mathcal{H}_A,\mathcal{H}_B,\mathcal{H}_C\cdots$ are all finite dimensional Hilbert space.
We use $S=\{A,B,C,\cdots\}$ to denote the set of the whole index of the subsystems.
For any $I\subset S$, for instance ~$A$ or $BC$ (containing systems $B$ and $C$), we use the marginal of $I$ to denote the joint state of subsystems $I$,
\begin{equation*}
\rho_{I} = \Tr_{{S}\setminus {I}}[\rho_{{S}}],
\end{equation*}
where $\rho_S$ denotes the state of the whole system $\mathcal{H}_{S}$, and we trace out the complementary system ${S}\setminus {I}$ of ${I}$.

Clearly, a set of $\rho_{I}$s need to fulfill certain conditions to make sure the existence of the golbal state $\rho_S$. This motivates the following defintion of quantum marginal problem
\begin{Def}\label{def:qmp}
A given family $\mathcal{K}\subset 2^S$ and quantum states $(\rho_{{I}})_{{I} \in \mathcal{K}}$ is called compatible if
there exists a density operator $\rho_{{S}}$
for the total system such that $\forall {I} \in \mathcal{K}$
\begin{equation*}
\rho_{{I}}= \mbox{Tr}_{{S} \setminus {I}}[\rho_{{S}}]\,.
\end{equation*}
The quantum marginal problem is to determine whether given family $\mathcal{K}\subset 2^S$ and quantum states $(\rho_{{I}})_{{I} \in \mathcal{K}}$ is compatible.
\end{Def}

To illustrate this definition, we study some examples.

Let $\cH_{J}=\mathcal{H}_{\mathcal{A}}\otimes \mathcal{H}_{\mathcal{B}}\otimes\cH_C$ and $\mathcal{K}=\{\{A,C\},\{A,B\}\}$. Given $\rho_{AC},\rho_{AB}$, we are asking the question that when this marginal problem is compatible. The answer for the general version of this problem is not clear yet, even only for the three qubits case.
\begin{Exam}
The answer for the classical version of this problem is always yes if $\rho_{A}:=\tr_C \rho_{AC}=\tr_B \rho_{AB}$. Notice that in this case, $\rho_{AC}$ and $\rho_{BC}$ are classical distributions, in other words, diagonal states. One can verify that $\rho_{ABC}(x,y,z)$ is compatible with $\rho_{AC}$ and $\rho_{BC}$, where the tripartite distribution $\rho_{ABC}(x,y,z)$  is defined by
\begin{equation*}
\rho_{ABC}(x,y,z)=\frac{\rho_{AC}(x,z)\rho_{AB}(x,y)}{\rho_{A}(x)}.
\end{equation*}

\end{Exam}
\begin{Exam}
For $\cH_{J}$ being three-qubit Hilbert space and $\rho_{AC}=\rho_{AB}$, the problem is indeed the two-qubit symmetric extension problem. Symmetric extendibility of bipartite states is of vital importance in
quantum information because of its central role in separability tests, one-way distillation of EPR pairs, one-way
distillation of secure keys, quantum marginal problems, and anti-degradable quantum channels.
This problem is solved in \cite{Chen2014} by proving $\tr(\rho_B^2)\geq \tr(\rho_{AB}^2)-4\sqrt{\det{\rho_{AB}}}$ is sufficient and necessary for symmetric extendibility.
\end{Exam}
\begin{Exam}
In \cite{Carlen2013}, necessary conditions according to Strong Subadditivity of Entropy \cite{Lieb1973,Lieb1975} are given,
\begin{eqnarray*}
S(AC)+S(AB)\geq S(A)+S(ABC)\geq S(A),\\
S(AC)+S(AB)\geq S(B)+S(C),
\end{eqnarray*}
where $S(\rho)=-\tr(\rho\log\rho)$ defines the entropy of the states.
\end{Exam}
\section{Quantum earth mover's distance}
In this section, we propose quantum earth mover's distance and show a no-go quantum Kantorovich-Rubinstein theorem.
\subsection{Classical earth mover's distance}
In the discrete version, a probabilistic coupling models two distributions with a single joint
distribution.
\begin{Def}
  Given $\mu, \nu$ distributions over finite or countably infinite $X$ and $X$, a
  distribution $\pi$ over pairs $X \times X$ is called a \emph{coupling}
  for $(\mu, \nu)$ if its marginal distributions $\pi_1 = \mu$ and $\pi_2= \nu$.
\end{Def}
In the above definition, the marginal distribution $\pi_1(x) = \sum_{y} \pi(x,y)$ and $\pi_2(y) = \sum_{x} \mu(x,y)$. We use $\mathfrak{P}(\mu,\nu)$ denotes the set of coupling of $\mu$ and $\nu$.
\begin{Def}
Let $X$ be a finite or countably infinite set, $h:X\times X\mapsto \mathbb{R}^{+}$ be a nonnegative distance function.
The earth mover's distance between distributions $\mu, \nu$ over $X$ is defined as
\begin{align*}
\inf_{\pi\in\mathfrak{P}(\mu,\nu)}\sum_{x,y\in X}h(x,y)\pi(x,y).
\end{align*}
\end{Def}

\subsection{Quantum earth mover's distance}
One new concept introduced in classical earth mover's distance is the probabilistic coupling.
The quantum counter part of probability distribution is the quantum state. We first provide a definition of quantum coupling of quantum states.
\begin{Def}
Given two density matrices $\rho_A$ and $\rho_B$ on $\mathcal{H}_A$ and $\mathcal{H}_B$ respectively,
let $\mathfrak{P}(\rho_A,\rho_B)$ denote the set of all density matrices $\rho_{AB}$ on $\mathcal{H}_A\otimes \mathcal{H}_B$
such that
\begin{align*}
\tr_B \rho_{AB} =\rho_A,\\
\tr_A \rho_{AB} =\rho_B\ .
\end{align*}
\end{Def}
$\mathfrak{P}(\rho_A,\rho_B)$ can be viewed as the set of {\em quantum couplings} of $\rho_A$ and $\rho_B$. Generally, quantum couplings are not unique. Different quantum couplings represent different ways to share quantumness between two quantum states.

$\mathfrak{P}(\rho_1,\rho_2)$ is never empty; it always contains $\rho_A\otimes \rho_B$.

One can directly observe that $\mathfrak{P}(\rho_1,\rho_2)$ is convex and compact. Parathasarathy \cite{Parthasarathy2005} and Rudolph \cite{Rudolph2004}
characterized $\mathfrak{P}(\rho_1,\rho_2)$ by identifying its extreme points. The quantum coupling has also be used to study the quantum entropies by Winter \cite{Winter2016}.

In the classical earth mover's distance, the distance function $h(x,y)$ is defined on the extreme points of the set $X\times X$. We will use a function defined on the the extreme points of bipartite system, the set of all bipartite pure states.

\begin{Def}
Let $\mathcal{H}$ be a finite dimensional Hilbert space, $h:\mathcal{H}\otimes \mathcal{H}\mapsto \mathbb{R}^{+}$ be a nonnegative distance function.
The quantum earth mover's distance between $\rho_1, \rho_2$ over $\mathcal{H}$ is defined as
\begin{align*}
\inf_{ {\rho_{1,2}\in\mathfrak{P}(\rho_1,\rho_2)}}\inf_{\rho_{1,2}=\sum_i p_i\op{\psi_i}{\psi_i}}\sum_{i}p_i h(\op{\psi_i}{\psi_i}).
\end{align*}
where the second infimum is taken over all finite decompositions.
\end{Def}

One can have the following
\begin{theorem}
For any continuous function $h$ on finite dimensional Hilbert space $\mathcal{H}\otimes \mathcal{H}$, the second infimum is attained, and the optimal
ensemble can be chosen to have $d^4+1$ elements, where $d$ is the dimension of $\mathcal{H}$.
\end{theorem}
\begin{proof}
 First, let us show that for any finite decomposition $\rho=\sum_{i=1}^np_i\op{\psi_i}{\psi_i}$,
we can provide a decomposition $\rho=\sum_{i=1}^{d^4+1} q_i\op{\phi_i}{\phi_i}$ with $d^4+1$ elements, such that
\be
\sum_{i=1}^np_i h(\op{\psi_i}{\psi_i})=\sum_{i=1}^{d^4+1} q_i h(\op{\phi_i}{\phi_i}).
\ee
To this, consider convex hull ${\cal A}$ of the set $\{ (\op{\psi_i}{\psi_i}, h(\op{\psi_i}{\psi_i}) )\}_{i=1}^n$.
One can conclude that $x=(\sum_{i=1}^n p_i \op{\psi_i}{\psi_i},\sum_{i=1}^n p_i h(\op{\psi_i}{\psi_i}))$ belongs ${\cal A}$.
The set ${\cal A}$ is a compact convex set, actually a polyhedron,
in $d^4$-dimensional real affine space (this comes from the fact that states belongs to the real $d^4$ dimensional
space of Hermitian operators and have unit trace). The set of extremal points is included
in the set $\{ (\op{\psi_i}{\psi_i}, h(\op{\psi_i}{\psi_i}) )\}_{i=1}^n$.
Then from Caratheodory theorem it follows that $x$ can be written as a convex combination of
at most $d^4+1$ extremal points, i.e. $x=\sum_{i_j} q_{i_j} (\op{\psi_{i_j}}{\psi_{i_j}}, h(\op{\psi_{i_j}}{\psi_{i_j}})$
where $j=1,\ldots d^4+1$. Writing $q_{i_j}=q_j$, $\op{\psi_{i_j}}{\psi_{i_j}}=\op{\phi_j}{\phi_j}$ we get
$\sum_{i=1}^n p_i \op{\psi_i}{\psi_i}=\sum_{j=1}^{d^4+1} q_j \op{\phi_j}{\phi_j}$ and $\sum_{i=1}^np_i h(\op{\psi_i}{\psi_i})=\sum_{j=1}^{d^4+1}q_j h(\op{\phi_j}{\phi_j})$.
 Thus we have found a decomposition that has $d^4+1$ elements, and returns the same value of average,
so that the infimum can be taken solely over such decompositions.
Then from continuity of the function and compactness of the set of states it follows that the infimum is attained.
\end{proof}

\subsection{Kantorovich-Rubinstein theorem}

Let $(X,d)$ be a metric space, let $\mathfrak{B}(X)$ denote the set of all Borel subsets of $X$. We are given two regular Borel measures $\mu,\nu$ on $X$. We use $\mathfrak{P}(\mu,\nu)$ to denote the set of regular probabilistic Borel measures $\pi$ on the topological product $X\times X$ such that for all $E\in\mathfrak{B}(X)$,
\begin{align*}
\mu(E)=\pi(E\times X),\\
\nu(E)=\pi(X\times E).
\end{align*}
In other words, $\mathfrak{P}(\mu,\nu)$ denotes the set of coupling of $\mu$ and $\nu$.

For $f:X\mapsto \mathrm{R}$, we define the expression $||f||_L$ by the equation
\begin{align*}
||f||_L=\sup \{\frac{|f(x)-f(y)|}{d(x,y)}:x,y\in X;x\neq y\}
\end{align*}
Then the Kantorovich-Rubinstein theorem states that, for compact space $(X,d)$, we have
\begin{align*}
\inf_{\pi\in\mathfrak{P}(\mu,\nu)}\int_{X\times X} d(x,y)\pi({\rm{dxdy}})=\sup\{\int_X f {\rm {d}}\mu-\int_X f\ {\rm{d}} {\nu}:||f||_L\leq 1\}.
\end{align*}

For discrete metric space $(X,d)$ with $d(x,y)=1$ if $x\neq y$.
\[
\sup\{\int_X fd\mu-\int_X fd\nu: ||f||_L\leq 1\}=D({\mu},{\nu})=\frac{1}{2}\sum_{x}|\mu(x)-\nu(x)|.
\]
The Kantorovich-Rubinstein theorem becomes its ``baby" version, see, e.g., \cite{Lindvall2002,Levin2009}.
  Let $\mu$ and $\nu$ be distributions over $X$ and let $\pi$ be a
  coupling. Then
  \[
    \frac{1}{2}\sum_{x}|\mu(x)-\nu(x)|=\inf_{\pi\in\mathfrak{P}(\mu,\nu)} \Pr_{(x,y) \sim \mu} [ x\neq y ] .
  \]
Based on this result, the so-called
\emph{coupling method} \cite{Aldous1983} is developed to show two
probabilistic processes converge by constructing a coupling that causes the
processes to become equal with high probability.

\subsection{No-go quantum Kantorovich-Rubinstein theorem}
In this subsection, we show a no-go about quantum generalization of Kantorovich-Rubinstein theorem.

Notice that for classical earth mover's distance, as well as the ``baby"  version of the Kantorovich-Rubinstein theorem, $\inf_{\pi\in\mathfrak{P}(\mu,\nu)}\sum_{x,y\in X}h(x,y)\pi(x,y)$ can be regarded as a linear programming, the infimum inner product of the distance vector and the probabilistic coupling.

A quantum generalization is to take the distance function $h:\mathcal{H}\otimes\mathcal{H}\mapsto \mathbb{R}^+$ as following
\[
h(\op{\psi}{\psi})=\tr(H\op{\psi}{\psi}).
\]
To make sure $h$ is a distance like positive function, we assume $H$ be semi-definite positive.

The Kantorovich-Rubinstein theorem, including its ``baby" version, can be interpreted as the following, the distance between two distributions is characterized by earth mover's distance.

We can actually show a no-go theorem about quantum generalization of Kantorovich-Rubinstein theorem.
\begin{theorem}\label{qKantorovich}
The quantum earth mover's distance can not characterize trace distance between quantum states. More precisely, there is no Hermitian $H$ and a bijection $f$ such that for any $\rho_1$ and $\rho_2$
\[
D(\rho_1,\rho_2)=f(\min_{ {\rho_{1,2}\in\mathfrak{P}(\rho_1,\rho_2)}}\tr(H\rho_{1,2}))
\]
\end{theorem}
If we remove the bijection restriction of $f$ and the choice of trace distance, then the statement is wrong: One can choose $H=P_{as}$ and $f:[0,1]\mapsto \{0,1\}$ such that $f(x)=0$ iff $x=0$. Naturally, this can characterize the discrete metric $d(\rho_1,\rho_2)=1$ if $\rho_1\neq \rho_2$, and $d(\rho,\rho)=0$.

Before proving this, we first study the bipartite marginal problem.

Consider a bipartite Hilbert space $\cH=\cH_A\otimes\cH_B$ with $d=d_{\cH_A}=d_{\cH_B}$. We study the relation on the marginal and symmetry of quantum states in $\cH$.

Before introducing our results, we start from the following notions.
$I_{\cH}$ is used to denote the identity operator of $\cH$.
$$S=\sum_{i,j=0}^{d-1}\op{ij}{ji}$$
denotes the SWAP operator in $\cH$ with the following property
$$S\ket{\alpha}_A\ket{\beta}_B=\ket{\beta}_A\ket{\alpha}_B$$
for all $\ket{\alpha}\in\cH_A$ and $\ket{\beta}\in\cH_B$.

Let $P_{s}=\frac{1}{2}(I+S)$ and $P_{as}=\frac{1}{2}(I-S)$ denote the projections onto the symmetric subspace and antisymmetric subspaces, respectively.

For any quantum state $\rho_{AB}$, we can define the following two probabilities,
\begin{Def}\label{def:probabilities}
\begin{align*}
\vec{p}(\rho_{AB})=(\tr(P_{as}\rho_{AB}),\tr(P_{s}\rho_{AB})),\\
\vec{q}(\rho_{AB})=(\tr(P_{s}\rho_{AB}),\tr(P_{as}\rho_{AB})).
\end{align*}
\end{Def}

Let the standard maximally entangled state be $\ket{\Phi}_{AB}=\frac{1}{\sqrt{d}}\sum_{i=0}^{d-1}\ket{i}\ket{i}$, we can write any pure state $\ket{\psi}\in\cH$ into the form
$$\ket{\psi}=(M\otimes I)=(I\otimes M^{T})\ket{\Phi}$$
with square matrix $M$ satisfying $||M||_2=1$.

Our first result is about the relation between the symmetric property of bipartite state and the fidelity between its marginals.
\begin{lemma}
For $\rho_{AB}\in \cD(\cH)$, we have
\begin{align*}
F(\rho_A,\rho_B)\geq |\tr(P_{as}\rho_{AB})-\tr(P_{s}\rho_{AB})|=D(\vec{p}(\rho_{AB}),\vec{q}(\rho_{AB})).
\end{align*}
\end{lemma}
\begin{proof}
The proof is divided into two steps. In the first step, we show this statement is valid for pure states. In the second step, we prove it holds for any state by the concavity arguments.

STEP 1: Let $\ket{\psi}_{AB}=(M\otimes I)\ket{\Phi}=(I\otimes M^{T})\ket{\Phi}$. By choosing $$E=\frac{M+M^T}{2},\ \ F=\frac{M-M^T}{2},$$ then we have
\begin{align*}
\tr(P_{s}\psi_{AB})=||EE^{\dag}||_1,\ \ \  \tr(P_{as}\psi_{AB})=||FF^{\dag}||_1.
\end{align*}
Since we are considering the fidelity between $\psi_A$ and $\psi_B$, they are regarded living in the same space now. We use $\ket{\psi}_{AB}$ to be the purification of $\psi_A$, and $S\ket{\psi}_{AB}$ to be the purification of $\psi_B$.

According to Uhlmann's theorem (Fact \ref{fac:Uhlmann}), we know that
\begin{align*}
F(\psi_A,\psi_B)=\max_{U}|\langle{\psi}|(U\otimes I)S\ket{\psi}|=\max_{U}|\tr(M^{\dag}UM^T)|=\max_{U}|\tr(M^{*}UM)|=||MM^{*}||_1,
\end{align*}
where $U$ is ranging over all unitary of $\cH_A$.

Notice that $$||MM^{*}||_1=||(MM^{*})^{\dag}||_1=||M^{T}M^{\dag}||_1,$$ we have
\begin{align*}
F(\psi_A,\psi_B)=&||MM^{*}||_1\\
                =&||\frac{1}{2}MM^{*}||_1+||\frac{1}{2}M^{T}M^{\dag}||_1\\
                \geq& ||\frac{1}{2}MM^{*}+\frac{1}{2}M^{T}M^{\dag}||_1\\
                =&||\frac{1}{2}(E+F)(E+F)^{*}+\frac{1}{2}(E+F)^{T}(E+F)^{\dag}||_1\\
                =&||\frac{1}{2}(E+F)(E^T-F^T)^{*}+\frac{1}{2}(E-F)(E+F)^{\dag}||_1\\
                =&||\frac{1}{2}(E+F)(E^{\dag}-F^{\dag})+\frac{1}{2}(E-F)(E^{\dag}+F^{\dag})||_1\\
                =& ||EE^{\dag}-FF^{\dag}||_1\\
                \geq &|||EE^{\dag}||_1-||FF^{\dag}||_1|\\
                =&|\tr(P_{as}\psi_{AB})-\tr(P_{s}\psi_{AB})|,
\end{align*}
where the first inequality is according to the triangle inequality, so is the second inequality.

STEP 2: Now we are going to show this statement is true for general quantum states.
Assume $\rho_{AB}=\sum p_i\psi_i$ with $\psi_i$s being pure states, we can have
\begin{align*}
\rho_A=\sum p_i\psi_{iA},\ \ \ \ \rho_B=\sum p_i\psi_{iB}.
\end{align*}
Therefore, we have
\begin{align*}
F(\rho_A,\rho_B)&=F(\sum p_i\psi_{iA},\sum p_i\psi_{iB})\\
&\geq \sum p_i F(\psi_{iA},\psi_{iB})\\
&\geq \sum p_i|\tr(P_{as}\psi_{i})-\tr(P_{s}\psi_{i})|\\
&\geq |\tr(P_{as}\sum p_i\psi_{i})-\tr(P_{s}\sum p_i\psi_{i})|\\
&= |\tr(P_{as}\rho_{AB})-\tr(P_{s}\rho_{AB})|,
\end{align*}
where the first inequality is due to the strong concavity of fidelity, the second is due to the pure state case, and the third inequality is because of the triangle inequality.
\end{proof}

We observe that this bound is tight by studying the following example.
\begin{Exam}
For any $0\leq \mu \leq 1$, let
$$
\ket{\psi}_{AB}=\sqrt{\frac{1-\mu}{{2}}}(\ket{00}+\ket{11})+\sqrt{\frac{\mu}{{2}}}(i\ket{01}-i\ket{10}).
$$
Then one can verify that
\begin{align*}
M=\frac{1}{\sqrt{2}}{\begin{bmatrix}\sqrt{1-\mu}&i\mu\\ -i\mu&\sqrt{1-\mu}\\ \end{bmatrix}},\\
\tr(P_{s}\psi_{AB})=1-\mu, \ \  \tr(P_{as}\psi_{AB})=\mu.
\end{align*}
Then
\begin{align*}
MM^*=\frac{1-2\mu}{2}I\Rightarrow ||MM^{*}||_1=|1-2\mu|.
\end{align*}
Then
\begin{align*}
F(\psi_A,\psi_B)&=||MM^{*}||_1=|\mu-(1-\mu)|=|\tr(P_{as}\rho_{AB})-\tr(P_{s}\rho_{AB})|.
\end{align*}
\end{Exam}

We can obtain the following relation between the symmetric property of a bipartite state and the distinguishability between its marginals.
\begin{lemma}
For $\rho_{AB}\in \cD(\cH)$, we have
\begin{align*}
D(\rho_A,\rho_B)\leq 2\sqrt{\tr(P_{as}\rho_{AB})\tr(P_{s}\rho_{AB})}=F(\vec{p}(\rho_{AB}),\vec{q}(\rho_{AB})).
\end{align*}
\end{lemma}
\begin{proof}
By the relation between fidelity and distance \ref{fact:fdrelation} and the above lemma, we have
\begin{align*}
D(\rho_A,\rho_B)\leq \sqrt{1-F^2(\rho_A,\rho_B)}\leq \sqrt{1-D^2(\vec{p}(\rho_{AB}),\vec{q}(\rho_{AB}))}=2\sqrt{\tr(P_{as}\rho_{AB})\tr(P_{s}\rho_{AB})}=F(\vec{p}(\rho_{AB}),\vec{q}(\rho_{AB})).
\end{align*}
\end{proof}

We observe that this bound is tight by studying the following example.
\begin{Exam}
For any $0\leq \mu \leq 1$, let
$$
\ket{\psi}_{AB}=\sqrt{\frac{1-\mu}{{2}}}(\ket{01}+\ket{10})+\sqrt{\frac{\mu}{{2}}}(\ket{01}-\ket{10}).
$$
Then one can verify that
\begin{align*}
E=\sqrt{\frac{1-\mu}{{2}}}{\begin{bmatrix}0&1\\ 1&0\\ \end{bmatrix}}, \tr(P_{s}\psi_{AB})=1-\mu,\\
F=\sqrt{\frac{\mu}{{2}}}{\begin{bmatrix}0&1\\ -1&0\\ \end{bmatrix}}, \tr(P_{as}\psi_{AB})=\mu.
\end{align*}
Then
\begin{align*}
EF^{\dag}+FE^{\dag}=\sqrt{[\mu(1-\mu)]}{\begin{bmatrix}-1&0\\ 0&1\\ \end{bmatrix}}.
\end{align*}
Then
\begin{align*}
D(\psi_A,\psi_B)&=||EF^{\dag}+FE^{\dag}||_1=2\sqrt{[\mu(1-\mu)]}=2\sqrt{\tr(P_{as}\rho_{AB})\tr(P_{s}\rho_{AB})}.
\end{align*}
\end{Exam}

Now we present the proof of Theorem \ref{qKantorovich}
\begin{proof}
Now the earth mover's distance between two quantum states $\rho_1,\rho_2$ satisfies
\begin{align*}
\inf_{ {\rho_{1,2}\in\mathfrak{P}(\rho_1,\rho_2)}}\inf_{\rho_{1,2}=\sum_i p_i\op{\psi_i}{\psi_i}}\sum_{i}p_i h(\op{\psi_i}{\psi_i})=\inf_{ {\rho_{1,2}\in\mathfrak{P}(\rho_1,\rho_2)}}\tr(H\rho_{1,2})=\min_{ {\rho_{1,2}\in\mathfrak{P}(\rho_1,\rho_2)}}\tr(H\rho_{1,2})
\end{align*}
for some $H\geq 0$.

We assume that there is some unitary invariant distance measure $d(\cdot,\cdot)$ such that for quantum states $\rho_1$, $\rho_2$, such that it is uniquely determined by their earth mover's distance. In other words, there is a bijection $f$ such that
\[
d(\rho_1,\rho_2)=f(\min_{ {\rho_{1,2}\in\mathfrak{P}(\rho_1,\rho_2)}}\tr(H\rho_{1,2})).
\]
For pure states $\ket{\alpha}$ and $\ket{\beta}$, their coupling is unique, $\ket{\alpha}\ket{\beta}$. As $d$ is unitary invariant, we have for any $U$
\[
d(\rho_1,\rho_2)=d(U\rho_1U^{\dag},U\rho_2U^{\dag}).
\]
That is
\[
f(\tr(H(\op{\alpha}{\alpha}\otimes\op{\beta}{\beta})))=f(\tr(H(U\op{\alpha}{\alpha}U^{\dag}\otimes U\op{\beta}{\beta}U^{\dag}))).
\]
According to the fact that $f$ is bijection, one can conclude that
\[
\tr(H(\op{\alpha}{\alpha}\otimes\op{\beta}{\beta}))=\tr(H(U\op{\alpha}{\alpha}U^{\dag}\otimes U\op{\beta}{\beta}U^{\dag}))=\tr((U^{\dag}\otimes U^{\dag})H(U\otimes U)(\op{\alpha}{\alpha}\otimes \op{\beta}{\beta})).
\]
That is for any $\ket{\alpha}$, $\ket{\beta}$ and unitary $U$
\[
\tr((H-(U^{\dag}\otimes U^{\dag})H(U\otimes U))(\op{\alpha}{\alpha}\otimes \op{\beta}{\beta}))=0
\]
Notice that $\op{\alpha}{\alpha}\otimes \op{\beta}{\beta}$ forms a basis of the space of the linear operators in $\mathcal{H}\otimes\mathcal{H}$, we have
\[
\tr((H-(U^{\dag}\otimes U^{\dag})H(U\otimes U))^2)=0
\]
Then,
\[
H=(U^{\dag}\otimes U^{\dag})H(U\otimes U)
\]
This is equivalent to the fact that $H$ is a linear combination of $P_s$ and $P_{as}$. In other words, there exist $\lambda_1,\lambda_2$ such that
\[
H=\lambda_1 I+\lambda_2 P_{as}
\]
Notice that the $I$ component is useless as we can shift the value of $\tr(H\rho_{1,2})$. By scaling, we only need to study $H=P_{as}$.

For any $0\leq\mu\leq \frac{1}{2}$, as Example 4, we choose
\begin{align*}
\rho_1=\frac{1}{2}I+\sqrt{[\mu(1-\mu)]}{\begin{bmatrix}-1&0\\ 0&1\\ \end{bmatrix}},\\
\rho_2=\frac{1}{2}I-\sqrt{[\mu(1-\mu)]}{\begin{bmatrix}-1&0\\ 0&1\\ \end{bmatrix}}.
\end{align*}
According to Lemma 8, we know that
\[
2\sqrt{\mu(1-\mu)}=D(\rho_1,\rho_2)\leq 2\sqrt{\tr(P_{as}\rho_{1,2})\tr(P_{s}\rho_{1,2})}
\]
Then we have
\[
\min_{ {\rho_{1,2}\in\mathfrak{P}(\rho_1,\rho_2)}}\tr(H\rho_{1,2})=\mu.
\]
As the function $f$ is assumed to be bijection, we know that for any $0\leq x\leq \frac{1}{2}$, the function must be in the form
\[
f(x)=2\sqrt{x(1-x)}.
\]
However, consider pure states $\rho_1=\op{\alpha}{\alpha}$ and $\rho_2=\op{\beta}{\beta}$, we have
\[
\min_{ {\rho_{1,2}\in\mathfrak{P}(\rho_1,\rho_2)}}\tr(H\rho_{1,2})=\frac{1-|\ip{\alpha}{\beta}|^2}{2}.
\]
The distance between $\rho_1=\op{\alpha}{\alpha}$ and $\rho_2=\op{\beta}{\beta}$ can not be determined by function $f(x)=2\sqrt{x(1-x)}$
\[
D(\op{\alpha}{\alpha},\op{\beta}{\beta})\neq \sqrt{(1-|\ip{\alpha}{\beta}|^2)(1+|\ip{\alpha}{\beta}|^2)}.
\]
This completes the proof.
\end{proof}

One can use similar idea to prove that the statement is true if we use infidelity $1-F(\rho_1,\rho_2)$ as the distance measure. Firstly, by Lemma 7 and Example 4, we know that if such function $f$ does exists, the the function must be
\[
f(x)=1-|1-2x|,
\]
for $0\leq x\leq \frac{1}{2}$.

However, for pure states $\rho_1=\op{\alpha}{\alpha}$ and $\rho_2=\op{\beta}{\beta}$, we have
\[
\min_{ {\rho_{1,2}\in\mathfrak{P}(\rho_1,\rho_2)}}\tr(H\rho_{1,2})=\frac{1-|\ip{\alpha}{\beta}|^2}{2}.
\]
The infidelity between $\rho_1=\op{\alpha}{\alpha}$ and $\rho_2=\op{\beta}{\beta}$ can not be determined by function $f(x)=1-|1-2x|$
\[
1-F(\op{\alpha}{\alpha},\op{\beta}{\beta})=1-|\ip{\alpha}{\beta}|\neq1-|1-2\frac{1-|\ip{\alpha}{\beta}|^2}{2}|=1-|\ip{\alpha}{\beta}|^2.
\]

\section{Quantum Kantorovich-Rubinstein Inequalities}
Although in equality version of quantum Kantorovich-Rubinstein theorem is not possible, we derive inequalities which can be regarded as quantum generalization of Kantorovich-Rubinstein theorem in this section. In particular, we show that
\begin{theorem}\label{Thm-fidelity}
For $\rho_A,\rho_B$ be quantum states with the same dimension,
$$
\frac{1+F^2(\rho_A,\rho_B)}{2}\leq \max_{\rho_{A,B}\in\mathfrak{P}(\rho_A,\rho_B)}\tr(P_s\rho_{AB})\leq \frac{1+F(\rho_A,\rho_B)}{2}
$$
For diagonal density operators $\Lambda_A$ and $\Lambda_B$ with the same dimension $d$, we can obtain a slightly different version
$$
F(\Lambda_A,\Lambda_B)+\frac{\min_{i}(\sqrt{\lambda_{A,i}}-\sqrt{\lambda_{B,i}})^2}{2}\leq \max_{\rho_{A,B}\in\mathfrak{P}(\rho_A,\rho_B)}\tr(P_s\rho_{AB})\leq \frac{1+F(\rho_A,\rho_B)}{2},
$$
where $\Lambda_A=\mathrm{diag}\{\lambda_{A,1},\cdots,\lambda_{A,d}\}$ and $\Lambda_B=\mathrm{diag}\{\lambda_{B,1},\cdots,\lambda_{B,d}\}$ .
\end{theorem}
The upper bound part follows from Lemma 7 directly. We notice that lower bound part of the diagonal version is not covered by the lower bound of the general version by studying the following example.

\begin{Exam}
Let $0<x<1$ and
\begin{align*}
\rho_A=\begin{bmatrix}\frac{1+x}{2}&0\\ 0&\frac{1-x}{2}\\ \end{bmatrix}
\rho_B=\begin{bmatrix}\frac{1-x}{2}&0\\ 0&\frac{1+x}{2}\\ \end{bmatrix}
\end{align*}
Then $F(\rho_A,\rho_B)=\sqrt{1-x^2}$ and
\begin{align*}
F(\Lambda_A,\Lambda_B)+\frac{\min_{i}(\sqrt{\Lambda_{Aii}}-\sqrt{\Lambda_{Bii}})^2}{2}=\sqrt{1-x^2}+\frac{1-\sqrt{1-x^2}}{2}=\frac{1+\sqrt{1-x^2}}{2}> \frac{2-x^2}{2}=\frac{1+F^2(\rho_A,\rho_B)}{2}.
\end{align*}
\end{Exam}

In the following, we first prove the lower bound of the diagonal version,
\begin{lemma}
\label{Lem-diag}
Given two distributions $(S_1,\cdots,S_n)$ and $(T_1,\cdots,T_n)$ (that is, $\sum_{i=1}^nS_i=\sum_{i=1}^nT_i=1$), there exists a $d\times d$ matrix $X$ with non-negative elements such that:
\begin{align*}
&\sum_{j=1}^nX_{ij} = S_i,  \quad\forall 1\le i\le n \\
&\sum_{i=1}^nX_{ij} = T_j, \quad\forall 1\le j\le n \\
\sum_{i>j} (\sqrt{X_{ij}}-\sqrt{X_{ji}}&)^2 \le \sum_i(\sqrt{S_i}-\sqrt{T_i})^2 - \min_{i}\{(\sqrt{S_i}-\sqrt{T_i})^2\}.
\end{align*}
\end{lemma}

\begin{proof}
We will show Algorithm \ref{algo1} produces such an $X$ satisfying the conditions in the lemma. The termination of the algorithm is ensured by the fact that, the size of $A$ is strictly decrease at each iteration of the while loop, and in for loop the size of $B$ is also bounded by $n$.

\begin{algorithm}[H]
\label{alg-find}
\caption{Algorithm for Lemma \ref{Lem-diag}}
\label{algo1}
\KwIn{$s_1,s_2,\cdots,s_n,t_1,t_2,\cdots,t_n$}
\KwOut{$n\times n$ matrix $X$}
Set $A := \{1,2,\cdots,n\}$\;
Set $X$ being a $n\times n$ matrix with zero entries\;
\While{$|A|>1$}{
    Choose $k\in A$ such that $\max\{s_k/t_k,t_k/s_k\} = \max_{i\in A}\{\max\{s_i/t_i,t_i/s_i\}\}$\;
    \uIf{$s_k = t_k$}
    {
        $X_{kk} := s_k$\;
        Delete $k$ from $A$\;
    }
    \uElseIf{$s_k > t_k$}
    {
        $B := \{i\in A, s_i<t_i\}$\;
        $m = s_k/t_k$\;
        \For{$i$ in $B$}
        {
            $x := (t_i-s_i)/(m-1)$\;
            \uIf{$x = t_k$}{
                $X_{ki} := s_k;\quad X_{ik} := t_k$\;
                $X_{ii} := t_i-s_k$\;
                Delete $i$ from $A$\;
                Delete $k$ from $A$\;
                {\bf Break}\;
            }
            \uElseIf{$x > t_k$}{
                $X_{ki} := s_k;\quad X_{ik} := t_k$\;
                $s_i := s_i-t_k;\quad t_i := t_i-s_k$\;
                Delete $k$ from $A$\;
                {\bf Break}\;
            }
            \Else{
                $X_{ki} := m\times x;\quad X_{ik} := x$\;
                $s_k := s_k-m\times x;\quad t_k := t_k-x$\;
                $X_{ii} := s_i-x$\;
                Delete $i$ from $A$\;
            }
        }
    }
    \Else{
        Similar to the ({\bf else if} {$s_k > t_k$}) part, but only replace all $s$ by $t$ and $t$ by $s$.
    }
}
$i := A[1]$; \tcp{Now, $A$ has only one element.}
$X_{ii} := s_i$\;
\Return X
\end{algorithm}

We prove the following statements by induction on the number of iterations.
\begin{align*}
{\bf statement: }&\text{ 1. at the beginning of each round,} \sum_{i\in A}s_i = \sum_{i\in A}t_i. \\
&\text{ 2. Each time update $s_i$ and $t_i$,}\  \sum_{j}X_{ij}+s_i = S_i,\ \ \sum_{j}X_{ji}+t_i = T_i. \\
&\text{ 3. if $i$ is deleted from $A$, then}\ \sum_{j}X_{ij} = S_i,\ \ \sum_{j}X_{ji} = T_i.
\end{align*}
At the beginning of the first iteration, statement 1, 2, 3 trivially hold. Now suppose at the beginning of some round, all statements are valid. Assume at line 4, $k\in A$ is chosen.

If $s_k = t_k$, then at line 7, we delete $k$ from $A$ to obtain $A^\prime$ which is just the set $A$ at the beginning of next iteration, and trivially
$$\sum_{i\in A^\prime}(s_i-t_i) = \sum_{i\in A}(s_i-t_i) - (s_k-t_k) = 0.$$
Moreover, line 6 $X_{kk}:= s_k$ is the first time being changed, so $\sum_iX_{ki} = \sum_{i\neq k}X_{ki}+X_{kk} = S_k-s_k+s_k = S_k$, and similarly, $\sum_iX_{ik}=t_k$.
Therefore, at the beginning of next iteration, all statements still hold.

Without lose of generality, we only analyse the case $s_k>t_k$. At first, each time we update $s_k$ and $t_k$, $s_k/t_k=m$ is always true because this update only happens in line 25, and trivially,
$$\frac{s_k-m\times x}{t_k-x} = \frac{m\times t_k-m\times x}{t_k-x} = m.$$
Suppose $x_i = (t_i-s_i)/(m-1)$. Using statement 1, it is not difficult to realize
$$
\sum_{i\in B}x_i = \frac{\sum_{i\in B}(t_i-s_i)}{s_k-t_k}t_k \ge \frac{s_k-t_k}{s_k-t_k}t_k = t_k,
$$
which implies that it is impossible that in for loop the program execute the line 25 - 27 for all $i\in B$; that is, after execute line 25 - 27 for the first several round, the program will always execute the first two choice of the if statement and then enter the next while loop round.

Each time the program executes line 25 - 27, statement 1 and 2 is not violated because before the updates:
$$s_i+s_k-t_i-t_k = -(m-1)x+s_k-t_k = s_k-mx-(t_k-x).$$
Moreover, the updated $s_k,t_k$ are all positive. The deletion of $i$ does not violate statement 3 because:
$$X_{ii}+X_{ik}=s_i, \quad X_{ii}+X_{ki}=s_i+(m-1)x=t_i, X_{ii} = \frac{ms_i-t_i}{m-1}\ge0,$$
where the last inequality is due to the choice of $k$ (that is, $m\ge t_i/s_i$).

If the program executes line 14 - 18, the validity of statement 1 and 3 is due to the fact:
$$
s_i+s_k-t_i-t_k = -(m-1)x+mt_k-t_k = 0.
$$
If the program executes line 20 - 23, statement 2 and 3 are easy to check and so the statement 1 is also valid.

In summary, all three statements hold during the loops. Therefore, after the program leaves the while loop, $s_i = t_i$ if $i\in A$, and so if we set $X_{ii} = s_i$, then
\begin{align*}
&\sum_{j=1}^nX_{ij} = S_i,  \quad\forall 1\le i\le n \\
&\sum_{i=1}^nX_{ij} = T_j, \quad\forall 1\le j\le n
\end{align*}
are both valid. The elements of $X$ are all non-negative which can be seen from above analysis.

To show that $X$ satisfies the third condition in lemma, we first show that each time we update some $s_i$ and $t_i$ to $s_i^\prime$ and $t_i^\prime$,
\begin{equation}
\label{eqn-s-t-dec}
(\sqrt{s_i^\prime}-\sqrt{t_i^\prime})^2\le(\sqrt{s_i}-\sqrt{t_i})^2.
\end{equation}
The update happens in line 21 and line 26 only. In line 21, note that $t_i-s_i=(m-1)x\ge(m-1)t_k$, so
$$t_i^\prime-s_i^\prime = (t_i-s_k)-(s_i-t_k) = t_i-s_i-mt_k+t_k\ge0,$$
and therefore using the fact $s_i\ge t_k$ and $m\ge\frac{t_i}{s_i}$,
$$
0\le \sqrt{t_i^\prime}-\sqrt{s_i^\prime} =
\sqrt{t_i -mt_k}-\sqrt{s_i-t_k}
\le \sqrt{t_i -\frac{t_i}{s_i}t_k}-\sqrt{s_i-t_k}
= (\sqrt{t_i}-\sqrt{s_i})\sqrt{1-\frac{t_k}{s_i}}
\le (\sqrt{t_i}-\sqrt{s_i}).
$$
In line 26, similarly we have
$$
\sqrt{s_i^\prime}-\sqrt{t_i^\prime} =
\sqrt{s_i -mx}-\sqrt{t_i-x} =
(\sqrt{s_i}-\sqrt{t_i})\sqrt{1-\frac{mx}{s_i}}
$$
where the factor $0\le\sqrt{1-\frac{mx}{s_i}}\le1$.
Finally, we show that at each round of while loop, if $k$ is chosen, then after this round,
$$
\sum_{i\neq k}(\sqrt{X_{ik}}-\sqrt{X_{ki}})^2\le (\sqrt{S_k}-\sqrt{T_k})^2.
$$
We assume $s = s_k$ and $t = t_k$ being the corresponding value at the beginning of this round, so due to Eqn. (\ref{eqn-s-t-dec}), $$(\sqrt{s}-\sqrt{t})^2\le(\sqrt{S_k}-\sqrt{T_k})^2.$$
If we realize the fact that, if $X_{ik}$ and $X_{ki}$ are updated, then
$$\frac{X_{ki}}{X_{ik}} = m = \frac{s}{t}, \quad\sum_{i\neq k} X_{ik} = t$$
therefore,
$$
\sum_{i\neq k}(\sqrt{X_{ik}}-\sqrt{X_{ki}})^2 = \sum_{i\neq k}(\sqrt{X_{ik}}-\sqrt{mX_{ik}})^2 = \sum_{i\neq k}X_{ik}(1-\sqrt{m})^2
= t(1-\sqrt{s/t})^2 = (\sqrt{t}-\sqrt{s})^2.
$$
The off-diagonal elements are only updated during the while loop, and only at most $n-1$ round of the iteration (suppose all the chosen $k$ form a set $K$, then $|K|\le n-1$), so:
\begin{align*}
\sum_{i>k} (\sqrt{X_{ik}}-\sqrt{X_{ki}})^2 &\le
\sum_{k\in K}\sum_{i\neq k}(\sqrt{X_{ik}}-\sqrt{X_{ki}})^2\le
\sum_{k\in K}(\sqrt{S_k}-\sqrt{T_k})^2 \\
&\le
\sum_{k}(\sqrt{S_k}-\sqrt{T_k})^2 - \min_{k}\{(\sqrt{S_k}-\sqrt{T_k})^2\}
\end{align*}
which complete the proof.
\end{proof}

\begin{proof}[Proof of the diagonal version of Theorem \ref{Thm-fidelity}]
Suppose $\Lambda_A = {\rm diag}\{S_1,S_2,\cdots,S_d\}= {\rm diag}\{\lambda_{A,1},\cdots,\lambda_{A,d}\}$ and $\Lambda_B = {\rm diag}\{T_1,T_2,\cdots,T_d\}= {\rm diag}\{\lambda_{B,1},\cdots,\lambda_{B,d}\}$. Let $X$ being the $n\times n$ matrix in Lemma \ref{Lem-diag}. Now, we construct the density operator $\rho$:
$$
\rho = \sum_i X_{ii}|i\>|i\>\<i|\<i| + \sum_{i>j}\Big(\sqrt{X_{ij}}|i\>|j\> +
\sqrt{X_{ji}}|j\>|i\>\Big)\Big(\sqrt{X_{ij}}\<i|\<j| + \sqrt{X_{ji}}\<j|\<i|\Big).
$$
It is easy to verify that $\rho\in\mathfrak{P}(\Lambda_A,\Lambda_B)$. Moreover,
\begin{align*}
\tr(P_s\rho) &= \sum_iX_{ii} + \frac{1}{2}\sum_{i>j}\Big(X_{ij}+X_{ji}+2\sqrt{X_{ij}X_{ji}}\Big) \\
&= 1 - \frac{1}{2}\sum_{i>j}\Big(\sqrt{X_{ij}}-\sqrt{X_{ji}}\Big)^2 \\
&\ge 1-\frac{1}{2}\bigg[\sum_i(\sqrt{S_i}-\sqrt{T_i})^2 - \min_{i}\{(\sqrt{S_i}-\sqrt{T_i})^2\}\bigg] \\
&= F(\Lambda_A,\Lambda_B)+\frac{\min_{i}(\sqrt{\lambda_{A,i}}-\sqrt{\lambda_{B,i}})^2}{2}.
\end{align*}
\end{proof}

To prove the lower bound for general $\rho_A,\rho_B$, we need the following lemmas,
\begin{lemma}
\label{Lem-Tech1}
Given two $d$-dimensional density operators $\rho$ and $\sigma$,
there exists decompositions:
$$
\rho = \sum_{i=1}^ds_i|u_i\>\<u_i|,\quad
\sigma = \sum_{i=1}^dt_i|v_i\>\<v_i|
$$
where $0\le s_i,t_i \le1$ and $|u_i\>$s $|v_i\>$s are unit vectors for all $1\le i\le d$, such that:
$$
\<v_i|v_j\>=\delta_{i,j}, \<v_i|u_i\>\ge 0, \quad F(\rho,\sigma) = \sum_{i=1}^d\sqrt{s_it_i}\<v_i|u_i\>. %\quad \<u_1|u_2\> = \<v_1|v_2\>.
$$
\end{lemma}
\begin{proof}
The existence can be constructed from the constructive proof of Uhlmann's theorem.
Suppose $|\psi\> = \sum_i\sqrt{s_i}|u_i\>|i\>$ and $|\phi\> = \sum_i\sqrt{t_i}|v_i\>|i\>$ being the purifications of $\rho$ and $\sigma$ such that
 $$
 F(\rho,\sigma)=\ip{\phi}{\psi}=\sum_{i=1}^d\sqrt{s_it_i}\<v_i|u_i\>
 $$
According to $F(\rho,\sigma) \ge \<\phi|\psi\>$, we have $\<v_i|u_i\>\geq 0$.
\end{proof}

\begin{lemma}
\label{Lem-Tech2}
Given rank 2 density operators $\rho = s_1|u_1\>\<u_1|+s_2|u_2\>\<u_2|$ and
$\sigma = t_1|v_1\>\<v_1|+t_2|v_2\>\<v_2|$ such that
\begin{align*}
F(\rho,\sigma)=\sqrt{s_1t_1}\<v_1|u_1\>+\sqrt{s_2t_2}\<v_2|u_2\>,\\
\end{align*}
There exists a coupling $\tau\in\mathfrak{P}(\rho,\sigma)$, such that:
$$
\tr(P_s\tau) \ge \frac{1}{2}+\frac{1}{2}F(\rho,\sigma)^2.
$$
\end{lemma}

\begin{proof}
Moreover, we can have
$$
\ip{v_1}{v_2}=0,
\<u_1|u_2\>=re^{-i\theta},
$$
with $r,\theta\in\mathbb{R}$.

Let us construct the density operator $\tau$ as follows:
\begin{align*}
\tau =\ &(\sqrt{s_1t_1}|u_1\>|v_1\> - \sqrt{s_2t_2}e^{2i\theta}|u_2\>|v_2\>)(\sqrt{s_1t_1}\<u_1|\<v_1| - \sqrt{s_2t_2}e^{-2i\theta}\<u_2|\<v_2|) \\
&+ (\sqrt{s_1t_2}|u_1\>|v_2\> + \sqrt{s_2t_1}|u_2\>|v_1\>)(\sqrt{s_1t_2}\<u_1|\<v_2| + \sqrt{s_2t_1}\<u_2|\<v_1|).
\end{align*}
It is straightforward to check $\tau\in\mathfrak{P}(\rho,\sigma)$:
\begin{align*}
\tr_2(\tau) =\ s_1t_1|u_1\>\<u_1| + s_2t_2|u_2\>\<u_2|+ s_1t_2|u_1\>\<u_1| + s_2t_1|u_2\>\<u_2| =\rho
\end{align*}
using $t_1+t_2 = 1$,
\begin{align*}
\tr_1(\tau) =\ &s_1t_1|v_1\>\<v_1| + s_2t_2|v_2\>\<v_2| - \sqrt{s_1t_1s_2t_2} (e^{-2i\theta}|v_1\>\<v_2|\<u_2|u_1\>+e^{2i\theta}|v_2\>\<v_1|\<u_1|u_2\>) \\
&+ s_1t_2|v_1\>\<v_1| + s_2t_1|v_2\>\<v_2| + \sqrt{s_1t_1s_2t_2}|v_1\>\<v_2|\<u_1|u_2\>+|v_2\>\<v_1|\<u_2|u_1\>) \\
=\ &s_1|u_1\>\<u_1| + s_2|u_2\>\<u_2| \\
=\ &\sigma
\end{align*}
using $s_1+s_2 = 1$, $e^{i\theta}\<u_1|u_2\> = e^{-i\theta}\<u_2|u_1\>$.

Moreover, we compute $\tr(S\tau)$ ($S$ is the SWAP operator):
\begin{align*}
\tr(S\tau) =\ &(\sqrt{s_1t_1}\<u_1|\<v_1| - e^{-2i\theta}\sqrt{s_2t_2}\<u_2|\<v_2|)(\sqrt{s_1t_1}|v_1\>|u_1\> - e^{2i\theta}\sqrt{s_2t_2}|v_2\>|u_2\>) \\
&+ (\sqrt{s_1t_2}\<u_1|\<v_2| + \sqrt{s_2t_1}\<u_2|\<v_1|)(\sqrt{s_1t_2}|v_2\>|u_1\> + \sqrt{s_2t_1}|v_1\>|u_2\>) \\
=\ &s_1t_1(\<u_1|v_1\>)^2 + s_2t_2(\<u_2|v_2\>)^2 - \sqrt{s_1t_1s_2t_2}[e^{2i\theta}\<u_1|v_2\>\<v_1|u_2\>+e^{-2i\theta}\<u_2|v_1\>\<v_2|u_1\>] \\
&+ s_1t_2(\<u_1|v_2\>)^2 + s_2t_1(\<u_2|v_1\>)^2 + \sqrt{s_1t_1s_2t_2}[\<u_1|v_1\>\<v_2|u_2\>+\<u_2|v_2\>\<v_1|u_1\>] \\
=\ &[\sqrt{s_1t_1}\<u_1|v_1\> + \sqrt{s_2t_2}\<u_2|v_2\>]^2
+ |\sqrt{s_1t_2}\<u_1|v_2\>e^{2i\theta} - \sqrt{s_2t_1}\<u_2|v_1\>|^2 \\
\ge\ &F(\rho,\sigma)^2.
\end{align*}

Therefore, we obtain the following inequality:
\begin{align*}
\tr(P_s\tau) &= \tr\Big(\frac{1}{2}(I+S)\tau\Big) = \frac{1}{2} + \frac{1}{2}\tr(S\tau) \ge \frac{1}{2} + \frac{1}{2}F(\rho,\sigma)^2.
\end{align*}
\end{proof}

\begin{lemma}
\label{Lem-Tech3}
Given two distributions $(s_1,s_2,\cdots,s_d)$ and $(t_1,t_2,\cdots,t_d)$, there exists a matrix $X$ with non-negative elements such that:
\begin{align*}
&\forall\ 1\le i\le d, \quad \sum_jX_{ij} = 1, \\
&\forall\ 1\le i,j\le d, \quad X_{ij}s_i+X_{ji}s_j = X_{ij}t_i+X_{ji}t_j.
\end{align*}
\end{lemma}
\begin{proof}
The following Algorithm \ref{alg-find2} produces such an $X$ satisfies all conditions. The correctness of the algorithm can be easily checked.

\begin{algorithm}[H]
\label{alg-find2}
\caption{Algorithm for Lemma \ref{Lem-Tech3}}
\label{algo2}
\KwIn{$s_1,s_2,\cdots,s_n,t_1,t_2,\cdots,t_n$}
\KwOut{$n\times n$ matrix $X$}
Initial sets $A,B := \emptyset$; \tcp{Two empty sets $A$ and $B$.}
Initial array $c_1,c_2,\cdots,c_n := 1$\;
Initial $n\times n$ matrix $X$ with zero elements\;
\For{$i = 1$ to $n$}
{
    \uIf{$s_i = t_i$}{
        $X_{ii} := s_i$\;
    }
    \uElseIf{$s_i > t_i$}{
        add $i$ to $A$\;
    }
    \Else{
        add $i$ to $B$\;
    }
}
\While{$|A|\ge1$}{
    Choose $k\in A\cup B$ such that $c_k|s_k-t_k| = \min_{i\in A\cup B}\{c_i|s_i-t_i|\}$\;
    \If{$k\in A$}
    {
        choose $i\in B$\;
    }
    \Else{
        choose $i\in A$\;
    }
    $X_{ki} := c_k$\;
    $X_{ik} := X_{ki}|s_k-t_k|/|s_i-t_i|$\;
    $c_i := c_i-X_{ik}$\;
    delete $k$ from $A$ and $B$\;
    \If{$c_i = 0$}{
        delete $i$ from $A$ and $B$\;
    }
}
\Return X
\end{algorithm}
\end{proof}

\begin{proof}[Proof of Theorem \ref{Thm-fidelity} for general $\rho_A$ and $\rho_B$]

According to Lemma \ref{Lem-Tech1} and Lemma \ref{Lem-Tech3}, we decompose $\rho$ and $\sigma$ as:
$$
\rho = \sum_{i=1}^ds_i|u_i\>\<u_i|,\quad
\sigma = \sum_{i=1}^dt_i|v_i\>\<v_i|.
$$
and assume $X$ is the corresponding matrix.
We first construct following matrices:
\begin{align*}
&\forall\ 1\le i\le d, \quad \rho_i = s_iX_{ii}|u_i\>\<u_i|, \quad \sigma_i = t_iX_{ii}|v_i\>\<v_i|, \quad \tr(\rho_i) = \tr(\sigma_i) = s_iX_{ii} = t_iX_{ii}; \\
&\forall\ 1\le i\neq j\le d, \quad \rho_{ij} = s_iX_{ij}|u_i\>\<u_i| + s_jX_{ji}|u_j\>\<u_j|, \quad \sigma_{ij} = t_iX_{ij}|v_i\>\<v_i| + t_jX_{ji}|v_j\>\<v_j|, \\
&\qquad\qquad\qquad\qquad \tr(\rho_{ij}) = \tr(\sigma_{ij}) = s_iX_{ij}+s_jX_{ji}.
\end{align*}
Trivially,
$$
\rho = \sum_i\rho_i+\sum_{i< j}\rho_{ij},\quad\sigma = \sum_i\sigma_i+\sum_{i< j}\sigma_{ij}.
$$
Moreover, using the properties of $X$ and Uhlmann's theorem (see Remark in Lemma \ref{Lem-Tech1}), we observe:
\begin{align*}
&\sum_iF(\rho_i,\sigma_i) + \sum_{i< j}F(\rho_{ij},\sigma_{ij}) \\
\ge\ &\sum_is_iX_{ii}\<u_i|v_i\> + \sum_{i< j} (\sqrt{s_iX_{ij}t_iX_{ij}}\<u_i|v_i\> + \sqrt{s_jX_{ji}t_jX_{ji}}\<u_j|v_j\>) \\
=\ &\sum_i\sqrt{s_it_i}X_{ii}\<u_i|v_i\> + \sum_{i}\sum_{j\neq i }\sqrt{s_it_i}X_{ij}\<u_i|v_i\> \\
=\ &\sum_i\sum_{j}\sqrt{s_it_i}X_{ij}\<u_i|v_i\> \\
=\ &\sum_i\sqrt{s_it_i}\<u_i|v_i\> \\
=\ &F(\rho,\sigma)
\end{align*}
We choose $\tau_i = s_iX_{ii}|u_i\>|v_i\>\<u_i|\<v_i|$ being the coupling of $\rho_i$ and $\sigma_i$ such that
$$
\tr(P_s\tau_{i}) = \frac{1}{2}s_iX_{ii} + \frac{1}{2}s_iX_{ii}(\<u_i|v_i\>)^2 = \frac{1}{2}\tr(\rho_{i})+\frac{1}{2}\frac{1}{\tr(\rho_{i})}F(\rho_{i},\sigma_{i})^2,
$$
and $\tau_{ij}$ being the coupling of $\rho_{ij}$ and $\sigma_{ij}$ which satisfy
$$\tr(P_s\tau_{ij})\ge\frac{1}{2}\tr(\rho_{ij})+\frac{1}{2}\frac{1}{\tr(\rho_{ij})}F(\rho_{ij},\sigma_{ij})^2$$
according to Lemma \ref{Lem-Tech2} as both $\rho_{ij}$ and $\sigma_{ij}$ are rank 2 matrices. Therefore, $$\tau = \sum_i\tau_i+\sum_{i< j}\tau_{ij}$$ is a coupling of $\rho$ and $\sigma$.

Now, we are ready to obtain:
\begin{align*}
\tr(P_s\tau) &= \sum_i\tr(P_s\tau_{i}) + \sum_{i< j} \tr(P_s\tau_{ij}) \\
&\ge \frac{1}{2}\Big[\sum_i\tr(\rho_{i})+\sum_{i< j}\tr(\rho_{ij})\Big] + \frac{1}{2}\bigg[\sum_i\frac{F(\rho_{i},\sigma_{i})^2}{\tr(\rho_{i})}+ \sum_{i< j}\frac{F(\rho_{ij},\sigma_{ij})^2}{\tr(\rho_{ij})}\bigg] \\
&\ge \frac{1}{2} + \frac{1}{2} \frac{\left[\sum_kF(\rho_{k},\sigma_{k}) + \sum_{i< j}F(\rho_{ij},\sigma_{ij})\right]^2}{\sum_k\tr(\rho_{k}) + \sum_{i< j}\tr(\rho_{ij})} \\
&\ge \frac{1}{2} + \frac{1}{2} F(\rho,\sigma)^2
\end{align*}
using Cauchy-Schwarz inequality.
\end{proof}

\section{Tripartite Marginal Problem}
In this section, we employ the techniques in Section 3 to study the tripartite marginal problem. Consider a tripartite Hilbert space $\cH=\cH_A\otimes\cH_B\otimes\cH_C$ with $d_1=d_{\cH_A}=d_{\cH_B}$ and $d_2=d_{\cH_C}$. We study the problem on the existence of a tripartite state $\rho_{ABC}\in\mathcal{D}(\cH)$ with given reduced density matrices $\rho_{AB}=\tr_{C}\rho_{ABC}\in\mathcal{D}(\cH_A\otimes\cH_B)$, $\rho_{AC}=\tr_{B}\rho_{ABC}\in\mathcal{D}(\cH_A\otimes\cH_C)$ and $\rho_{BC}=\tr_{A}\rho_{ABC}\in\mathcal{D}(\cH_B\otimes\cH_C)$.

Another necessary condition of the marginal problem which gives a constrain on the fidelity of its marginals.
\begin{theorem}
For $\rho_{ABC}\in \cD(\cH)$, we have
\begin{align*}
F(\rho_{AC},\rho_{BC})\geq|\tr[(P_{as}\otimes I_C)\rho_{ABC}]-\tr[(P_{s}\otimes I_C)\rho_{ABC}]|= |\tr(P_{as}\rho_{AB})-\tr(P_{s}\rho_{AB})|=D(\vec{p}(\rho_{AB}),\vec{q}(\rho_{AB})).
\end{align*}
where $\vec{p}(\rho_{AB}),\vec{q}(\rho_{AB})$ are defined in \ref{def:probabilities}.
\end{theorem}
\begin{proof}
The proof is also divided into two steps. In the first step, we show this statement is valid for pure states. In the second step, we prove it holds for any state by the concavity arguments.

STEP 1: Let $$\ket{\psi}_{ABC}=\sum_{j=0}^{d_2-1}(M_j\otimes I)\ket{\Phi}_{AB}\ket{j}_C=\sum_{j=0}^{d_2-1}(I\otimes M_j^T)\ket{\Phi}_{AB}\ket{j}_C.$$

We choose
\begin{align*}
E_j=\frac{M_j+M_j^T}{2},\ \ F_j=\frac{M_j-M_j^T}{2} \Longrightarrow M_j=E_j+F_j,\ \ E_j^T=E_j,\ \ \ F_j^T=-F_j.
\end{align*}
We can verify the following,
\begin{align*}
&\tr[(P_{s}\otimes I_C)\ \psi_{ABC}]=\tr(P_{s}\psi_{AB})=\sum_{j=0}^{d_2-1}||E_j^{\dag}E_j||_1,\\
&\tr[(P_{as}\otimes I_C)\ \psi_{ABC})]=\tr(P_{as}\psi_{AB})=\sum_{j=0}^{d_2-1}||F_j^{\dag}F_j||_1.
\end{align*}

Since we are considering the fidelity between $\psi_{AC}$ and $\psi_{BC}$, they are regard living in the same space now. We use $\ket{\psi}_{ABC}$ to be the purification of $\psi_{AC}$, and $S_{AB}\ket{\psi}_{ABC}$ to be the purification of $\psi_{BC}$.

According to Uhlmann's theorem (Fact \ref{fac:Uhlmann}), we know that
\begin{align*}
F(\psi_{AC},\psi_{BC})=\max_{U_A}|\langle{\psi}|(U_A\otimes I_{BC})S_{AB}\ket{\psi}|=\max_{U_A}|\sum_{j=0}^{d_2-1}\tr(M_j^{\dag}U_AM_j^T)|=||\sum_{j=0}^{d_2-1}M_jM_j^{*}||_1,
\end{align*}
where $S_{AB}$ is the SWAP operator of $\cH_A\otimes\cH_B$ and $U_A$ is ranging over all unitary of $\cH_A$.

Notice that $$||\sum_{j=0}^{d_2-1}M_jM_j^{*}||_1=||(\sum_{j=0}^{d_2-1}M_jM_j^{*})^{\dag}||_1=||\sum_{j=0}^{d_2-1}M_j^{T}M_j^{\dag}||_1,$$ we have
\begin{align*}
F(\psi_{AC},\psi_{BC})=&||\sum_{j=0}^{d_2-1}M_jM_j^{*}||_1\\
                =&||\frac{1}{2}\sum_{j=0}^{d_2-1}M_jM_j^{*}||_1+||\frac{1}{2}\sum_{j=0}^{d_2-1}M_j^{T}M_j^{\dag}||_1\\
                \geq& ||\frac{1}{2}\sum_{j=0}^{d_2-1}M_jM_j^{*}+\frac{1}{2}\sum_{j=0}^{d_2-1}M_j^{T}M_j^{\dag}||_1\\
                =&||\frac{1}{2}\sum_{j=0}^{d_2-1}[(E_j+F_j)(E_j+F_j)^{*}+(E_j+F_j)^{T}(E_j+F_j)^{\dag}]||_1\\
                =&||\frac{1}{2}\sum_{j=0}^{d_2-1}[(E_j+F_j)(E_j^T-F_j^T)^{*}+(E_j-F_j)(E_j+F_j)^{\dag}]||_1\\
                =&||\frac{1}{2}\sum_{j=0}^{d_2-1}[(E+F)(E^{\dag}-F^{\dag})+(E_j-F_j)(E_j^{\dag}+F_j^{\dag})]||_1\\
                =& ||\sum_{j=0}^{d_2-1}(E_jE_j^{\dag}-F_jF_j^{\dag})||_1\\
                \geq &|||\sum_{j=0}^{d_2-1}E_jE_j^{\dag}||_1-||\sum_{j=0}^{d_2-1}F_jF_j^{\dag}||_1|\\
                =& |\tr[(P_{as}\otimes I_C]\rho_{ABC}]-\tr[(P_{s}\otimes I_C)\rho_{ABC}]|\\
                =&|\tr(P_{as}\psi_{AB})-\tr(P_{s}\psi_{AB})|,
\end{align*}
where the first inequality is according to the triangle inequality, so is the second inequality.

STEP 2: Now we are going to show this statement is true for general quantum states.
Assume $\rho_{ABC}=\sum p_i\psi_i$ with $\psi_i$s being pure states, we can have
\begin{align*}
\rho_{AC}=\sum p_i\psi_{iAC},\ \ \ \ \rho_{BC}=\sum p_i\psi_{iBC}.
\end{align*}
Therefore, we have
\begin{align*}
F(\rho_{AC},\rho_{BC})&=F(\sum p_i\psi_{iAC},\sum p_i\psi_{iBC})\\
&\geq \sum p_i F(\psi_{iAC},\psi_{iBC})\\
&\geq \sum p_i|\tr(P_{as}\psi_{i})-\tr(P_{s}\psi_{i})|\\
&\geq |\tr(P_{as}\sum p_i\psi_{i})-\tr(P_{sy}\sum p_i\psi_{i})|\\
&= |\tr[(P_{as}\otimes I_C]\rho_{ABC}]-\tr[(P_{s}\otimes I_C)\rho_{ABC}]|\\
&= |\tr(P_{as}\rho_{AB})-\tr(P_{s}\rho_{AB})|,
\end{align*}
where the first inequality is due to the strong concavity of fidelity, the second is due to the pure state case, and the third inequality is because of the triangle inequality.
\end{proof}

We observe that this bound is tight by studying Example 4.

We first show that the following necessary condition of the above marginal problem which gives a constrain on the distance of its marginal.
\begin{theorem}
For $\rho_{ABC}\in \cD(\cH)$, we have
\begin{align*}
D(\rho_{AC},\rho_{BC})\leq 2\sqrt{\tr[(P_{as}\otimes I_C)\ \rho_{ABC}]\tr[(P_{s}\otimes I_C)\ \rho_{ABC}]}=2\sqrt{\tr(P_{as}\rho_{AB})\tr(P_{s}\rho_{AB})}=F(\vec{p}(\rho_{AB}),\vec{q}(\rho_{AB})),
\end{align*}
where $P_{as}$ and $P_{s}$ are respectively defined as the projections onto the anti-symmetric subspace and the symmetric subspace of $\cH_A\otimes\cH_B$, and $\vec{p}(\rho_{AB}),\vec{q}(\rho_{AB})$ are defined in \ref{def:probabilities}.
\end{theorem}
\begin{proof}
By the relation between fidelity and distance \ref{fact:fdrelation} and the above lemma, we have
\begin{align*}
D(\rho_{AC},\rho_{BC})\leq \sqrt{1-F^2(\rho_{AC},\rho_{BC})}\leq \sqrt{1-D^2(\vec{p}(\rho_{AB}),\vec{q}(\rho_{AB}))}=2\sqrt{\tr(P_{as}\rho_{AB})\tr(P_{s}\rho_{AB})}=F(\vec{p}(\rho_{AB}),\vec{q}(\rho_{AB})).
\end{align*}
\end{proof}

One can directly employ Example 5 to observe this bound is tight.

Based on these two results, we are able to give a class of criteria for the tripartite marginal problem.
\begin{theorem}
For $\rho_{AB}\in\mathcal{D}(\cH_A\otimes\cH_B)$, $\rho_{AC}\in\mathcal{D}(\cH_A\otimes\cH_C)$ and $\rho_{BC}\in\mathcal{D}(\cH_B\otimes\cH_C)$, there exists a tripartite state $\rho_{ABC}\in\mathcal{D}(\cH_A\otimes\cH_B\otimes \cH_C)$ with these marginals only if
\begin{align*}
F[\rho_{AC},((\mathcal{E}_B\otimes I_C)\rho_{BC}]\geq |\tr[P_{as}(I_A\otimes \mathcal{E}_B)(\rho_{AB})]-\tr[P_{s}(I_A\otimes \mathcal{E}_B)(\rho_{AB})|
%,\\
%D[\rho_{AC},((\mathcal{E}_B\otimes I_C)\rho_{BC}]\leq 2\sqrt{\tr P_{as}(I_A\otimes \mathcal{E}_B)(\rho_{AB})]\tr[P_{s}(I_A\otimes \mathcal{E}_B)(\rho_{AB})}
\end{align*}
hold for any quantum channel $\mathcal{E}_B$.
\end{theorem}
\begin{proof}
If there exists such a $\rho_{ABC}$ for $\rho_{AB}$, $\rho_{BC}$ and $\rho_{AC}$, then for any unitary quantum channel $\mathcal{E}_B$, the state $(I_{AC}\otimes\mathcal{E}_B)(\rho_{ABC})$ must have marginals $(I_A\otimes\mathcal{E}_B)\rho_{AB}$, $(I_C\otimes\mathcal{E}_B)\rho_{BC}$ and $\rho_{AC}$. Then we can directly apply Theorem 16.
\end{proof}

To see our criteria are ``universal", we notice that for any $\cH_{A}$, $\cH_B$ and $\cH_C$ with dimension $d_A\leq d_B\leq d_C$, the system can be embedded into $d_B\otimes d_B\otimes d_C$, or $d_A\otimes d_C\otimes d_C$, or $d_C\otimes d_B\otimes d_C$. Now two local dimensions are equal and we can apply our results.

\section{Conclusion and Discussion}
In this paper, we propose a quantum earth mover's distance and prove a no-go quantum
Kantorovich-Rubinstein theorem. Then we provide a new class of criteria for the tripartite marginal problem.

It would be interesting to find more applications of this quantum earth mover's distance. One of the possible application is to derive other versions of quantum
Kantorovich-Rubinstein theorem. For marginal problem, one problem is to generalize our criteria into the multipartite version.
Recall that in our proof, we use projections onto the symmetric subspace and anti-symmetric subspace which are nothing but the two copy irreducible representations induced by Shur-Weyl duality. We believe that the multi-copy irreducible representations given in the general Shur-Weyl duality \cite{Fulton2004} play a key role in the studying of the multipartite quantum marginal problem.

Theory of classical optimal transpose has being increasingly used to unlock various problems in imaging sciences (such as color or texture processing), computer vision and graphics (for shape manipulation) or machine learning (for regression, classification and density fitting).
We believe that the theory of quantum optimal transpose would be useful in studying machine learning, in particular, quantum machine learning.

We thank Doctor Christian Schilling and Professor Andreas Winter for suggestions on references. We are grateful for Professor Andreas Winter's discussion about Kantorovich-Rubinstein theorem.

This work is supported by DE180100156.

\newcommand{\etalchar}[1]{$^{#1}$}

\vspace*{-1ex}

\clearpage

\appendix

\end{document}